\documentclass[a4paper]{amsproc}
\usepackage{amssymb}
\usepackage{articletemplate_sde}

\usepackage{mathrsfs} 

\title[Ruin probabilties]{Ruin probabilities under general investments and heavy-tailed claims}

\subjclass[2000]{60F10; 60H20}

\author[Hult]{\bfseries Henrik Hult}

\address{
Department of Mathematics \\ 
KTH\\ 
STOCKHOLM\\
SWEDEN}
\email{hult@kth.se}

\author[Lindskog]{\bfseries Filip Lindskog}

\address{
Department of Mathematics \\ 
KTH\\ 
STOCKHOLM\\
SWEDEN}
\email{lindskog@kth.se}

\begin{document}

{\begin{flushleft}\baselineskip9pt\scriptsize
September 25, 2008 
\end{flushleft}}
\vspace{18mm}
\setcounter{page}{1}
\thispagestyle{empty}

\begin{abstract}
  In this paper we study the asymptotic decay of finite time ruin 
  probabilities for an insurance company that 
  faces heavy-tailed claims, uses predictable 
  investment strategies and makes investments in risky assets whose 
  prices evolve according to quite general semimartingales. 
  We show that the ruin problem corresponds to determining 
  hitting probabilities for the solution to a randomly perturbed 
  stochastic integral equation. 
  We derive a large deviation result for the hitting probabilities that
  holds uniformly over a family of semimartingales and show that this 
  result gives the asymptotic decay of finite time ruin 
  probabilities under arbitrary investment strategies, including 
  optimal investment strategies.
\end{abstract}

\maketitle

\section{Introduction}

Consider the following model for the evolution of the risk reserve of
an insurance company. The cumulative premiums minus claims up to time
$t$ are modeled by a L\'evy process, denoted $\vep Y_t$, whose downward jumps
are assumed to have a heavy-tailed (regularly varying)
distribution. 
The insurance company has the opportunity to deposit its capital to a 
bank account 
giving instantaneous interest rate $r_t$ and to invest 
its capital by taking positions in $n$ risky assets with spot prices
$S^k_t$, $k=1,\dots,n$.
We assume that the spot prices form strictly positive semimartingales 
and that the interest rates form a \cadlag\ adapted process.
We let
$\pi^0_t$ denote the fraction of the risk reserve deposited to the
bank account and let, for $k=1,\dots,n$, $\pi^k_t$ 
denote the fraction 
invested in the $k$th risky asset at time $t$. It is assumed that
$\pi_t = (\pi^0_t,\dots, \pi^n_t)$ is a \caglad~ predictable
process. By construction $\pi^0_t +  
\dots + \pi^n_t = 1$. 
With this notation the evolution of the risk reserve $X^{\vep}_t$ over
time is given by 
the stochastic integral equation
\begin{align}\label{eq:intro1}
  X^{\vep}_t = x + \int_{0+}^t \pi^0_{s} X^{\vep}_{s-}r_{s-} ds +
  \sum_{k=1}^n\int_{0+}^t \pi^k_{s} X^\vep_{s-}\frac{dS^k_{s}}{S^k_{s-}} 
  + \vep Y_t, 
  \quad t \geq 0,
\end{align}
where $x > 0$ denotes the initial capital.
In this paper the ruin probability over a finite time interval, which
without loss of generality is taken to be $[0,1]$, is studied. 
Since this probability cannot be computed without assuming a particular 
(simple) parametric model, we will rely on asymptotic approximations. 
In this paper the asymptotic decay of the ruin probability
$\Prob(\inf_{t\in [0,1]}X^{\vep}_t < 0)$ is determined, as $\vep \to 0$. 
The investment strategies are allowed to depend on $\vep$; natural examples
would be strategies that are functions of the reserve- and
premium-minus-claims processes.
Moreover, the asymptotic decay of the ruin probability under
optimal investment strategies is obtained
(see Theorem \ref{thm:hittingprob} and Corollary \ref{thm:asoptimal2}). 

The formulation of the ruin problem can be restated in terms of
hitting probabilities for the solution to the stochastic integral equation
 \begin{align}\label{eq:rr}
  X^{\vep}_t = x + \int_{0+}^t X^{\vep}_{s-}dZ_s + \vep Y_t, 
  \quad t \in [0,1],
\end{align}
where $Z$ is a semimartingal. In particular, the stochastic integral
equations \eqref{eq:intro1} and \eqref{eq:rr} coincide on $[0,1]$ if
\begin{align}\label{eq:Zpi}
  Z_t = \int_{0+}^t \pi^0_{s}r_{s-} ds +
  \sum_{k=1}^n\int_{0+}^t \pi^k_{s} \frac{dS^k_{s}}{S^k_{s-}}.
\end{align}
If the quadratic covariation
process $[Z,Y] = 0$ a.s.~it follows from 
It\^o's formula (see Lemma \ref{lem:sol}) that the solution 
$X^{\vep}$ to \eqref{eq:rr} is given by 
\begin{align}\label{eq:rr3}
  X^\vep_t = \calE(Z)_{t} \left(x + \vep
  \int_{0+}^t \frac{dY_s}{\calE(Z)_{s-}}\right),
  \quad t \in [0,1],
\end{align}
and $X^0_t = x \calE(Z)_t$, where $\calE(Z)$ denotes the 
\emph{Dol\'eans-Dade exponential} (\cite{P04}, p.~84)
\begin{align*}
  \calE(Z)_t = e^{Z_t - \frac{1}{2}[Z,Z]^c_t}
  \prod_{s \in (0,t]}(1+\Delta Z_s)e^{-\Delta Z_s}.
\end{align*}
Here $[Z,Z]^c$ is the continuous part of the quadratic variation
process and $\Delta Z_t = Z_t - Z_{t-}$. Note that if $Z$ has jumps bounded
below by $-1$ and $\inf_{t \in (0,1]} \Delta Z_t > -1$,
then $\calE(Z)_t$ is strictly positive and it follows that 
$\inf_{t \in [0,1]} X^0_t > 0$. However,  the process $\vep Y$ may cause 
$X^\vep_t$ to be negative but as $\vep \to 0$ such events become more
and more rare. Using a functional large deviation result for
stochastic integrals driven by regularly varying L\'evy processes the
asymptotic decay of the hitting  
probability  $\Prob(\inf_{t\in [0,1]}X^{\vep}_t < 0)$ as $\vep\to 0$,
is obtained (under a natural moment condition on $\calE(Z)$). 
This immediately gives the asymptotic decay of 
the ruin probability.

Letting $\vep \to 0$ in the ruin problem means that we are
studying the decay of the ruin probability when the 
premiums-minus-claims process becomes (arbitrary) small compared to
the risk reserve. Alternatively, one can keep $\vep$ fixed and let the initial
capital $x \to \infty$. This is the more popular approach in the risk
theory literature. From \eqref{eq:rr3} we see that
\begin{align*}
  \Prob\left(\inf_{t\in [0,1]}X^{\vep}_t < 0\right) 
  &= \Prob\left(\inf_{t \in [0,1]} 
  \left\{x\!+\! \vep \!\int_{0+}^t \!\frac{dY_s}{\calE(Z)_{s-}}
  \right\}<0\right)\\
  &= \Prob\left(\inf_{t \in [0,1]} 
  \int_{0+}^t \!\frac{dY_s}{\calE(Z)_{s-}} < -\frac{x}{\vep} \right) 
\end{align*}
and hence the asymptotic analysis in the two cases is identical. 

Of particular interest is the
asymptotic decay of the ruin probability under 
an {\it optimal} investment strategy; i.e.~a
strategy that minimizes the ruin probability. We prove a large
deviation result for hitting probabilities for $X^\vep$ in
\eqref{eq:rr3} with $Z$ as in \eqref{eq:Zpi} which holds uniformly over
a family $\Pi$ of investment strategies $\pi$:  
\begin{align}\label{eq:optstr}
    \lim_{\vep\to 0}\inf_{\pi\in\Pi}
    \frac{\Prob(\inf_{t\in [0,1]}X^{\vep,Z}_t < 0)}{\nu(-\infty,-\vep^{-1})}
    = x^{-\alpha}\inf_{\pi\in\Pi}\int_{0}^1\E\calE(Z)_{t}^{-\alpha} dt,
\end{align}
where $\nu$ is the L\'evy measure of $Y_1$.
Roughly speaking our result says
that, for small $\vep$, the optimal strategy (which may depend on $\vep$)
does not yield much
smaller ruin probability than, what we call, an 
{\it asymptotically optimal strategy}. That is, 
a strategy that minimizes the integral on the right-hand side 
in \eqref{eq:optstr}.
This is relevant, because
finding asymptotically optimal strategies is much easier than finding
optimal strategies. In some cases an
asymptotically optimal strategy can be explicitly calculated
(see Proposition \ref{thm:asoptimal} below).

In the special case where the asset price follows a geometric Brownian motion
and the premiums-minus-claims process is a compound Poisson process, 
the optimal investment
strategy, for the infinite time horizon ruin problem with interest rate 
$r = 0$, is
characterized in \cite{HP00}. There the authors use stochastic
control theory to characterize the optimal strategy as a solution to a
partial differential equation. In the case of heavy-tailed claim sizes, 
the asymptotic value (as the initial capital $x\to\infty$) 
of the optimal fraction invested in
the risky asset is determined in \cite{GG02} and \cite{S05}. It
coincides with the asymptotically optimal strategy (in the finite time 
horizon case) determined in Example \ref{ex:gbm2} below. 
When the asset price follows an exponential L\'evy process the
asymptotic decay of the ruin probability for constant investments
$\pi$ was recently studied in \cite{KK08}; also in the case of an
infinite time horizon.

The paper is organized as follows. Section \ref{sec:stoch} is devoted
to the asymptotic decay, as $\vep \to 0$, of hitting probabilities for
the solution $X^\vep$ in \eqref{eq:rr}. 
This result is applied to finite time horizon ruin
problems in Section \ref{sec:ruin}, where we also consider asymptotically 
optimal strategies. All the proofs and some auxiliary results are given
in Section \ref{sec:proofs}.

Throughout the paper we refer to \cite{P04} for definitions
and notation. We assume that all the random elements considered are 
defined on a complete filtered probability space 
$(\Omega,\mathcal{F},(\mathcal{F}_t)_{t\in [0,1]},\Prob)$ satisfying 
the \emph{usual hypotheses}, see p.~3 in \cite{P04}.

\section{Hitting probabilities for the solution to a stochastic integral
equation}\label{sec:stoch}

In this section we investigate hitting probabilities for the solution 
to a stochastic integral equation that is perturbed by small but heavy-tailed
random noise. The main result is Theorem \ref{thm:hittingprob} that
gives a large deviation result for hitting probabilities which holds
uniformly over a family of semimartingales. 
 
Consider the stochastic integral equation
\begin{align}\label{eq:rr2}
  X^{\vep}_t = x + \int_{0+}^t X^{\vep}_{s-}dZ_s + \vep Y_t, 
  \quad t \in [0,1],
\end{align}
where $Y$ is
a L\'evy process and 
$Z$ is a semimartingal. If the quadratic covariation
process $[Z,Y] = 0$ a.s.~it follows from 
It\^o's formula (see Lemma \ref{lem:sol} below) that the solution 
$X^{\vep}$ to \eqref{eq:rr2} is given by 
\begin{align}\label{eq:rr23}
  X^\vep_t = \calE(Z)_{t} \left(x + \vep
  \int_{0+}^t \frac{dY_s}{\calE(Z)_{s-}}\right),
  \quad t \in [0,1],
\end{align}
and $X^0_t = x \calE(Z)_t$. Suppose $Z$ has jumps bounded below by $-1$,
i.e.~$\inf_{t \in (0,1]}\Delta Z_t > -1$. 
Then $\calE(Z)_t$ is strictly positive 
and it follows that $\inf_{t \in [0,1]} X^0_t > 0$. However, for $\vep> 0$ the
process $Y$ may cause $X^\vep_t$ to take negative values and as $\vep \to 0$
this event becomes more and more rare. We are concerned with the
asymptotic decay of the probability that 
$\inf_{t \in [0,1]} X^\vep_t<0$. Using the explicit solution \eqref{eq:rr23}
it follows that 
\begin{align}\label{eq:events}
  \left\{\inf_{t \in [0,1]} X^\vep_t < 0\right\} 
  = \left\{\inf_{t \in [0,1]}
  \int_{0+}^t \frac{dY_s}{\calE(Z)_{s-}} < - \frac{x}{\vep}\right\}. 
\end{align}
Hence, it is sufficient to consider hitting probabilities for the stochastic
integral on the right hand side. 

Suppose, for now, that the L\'evy measure $\nu$ of $Y_1$ is regularly
varying. That is, there is an $\alpha > 0$ and a $p \in [0,1]$ such that, 
for all $\lambda > 0$,
\begin{align}\label{eq:levymcond0}
  \lim_{u\to\infty}
  \frac{\nu(-\infty,-\lambda u)}{\nu(-\infty,-u)\cup(u,\infty)} 
  = p\lambda^{-\alpha},\quad 
  \lim_{u\to\infty}
  \frac{\nu(\lambda u,\infty)}{\nu(-\infty,-u)\cup(u,\infty)}
  = (1-p)\lambda^{-\alpha}.
\end{align}
Using \eqref{eq:events} together with 
a functional large deviation result in \cite{HL07} for 
stochastic integral processes driven by regularly 
L\'evy processes, the asymptotic decay of the hitting probability can
be obtained. A modification of Example 3.2 in \cite{HL07} is
the following. 
\begin{prop}\label{thm:old}
 Let $Y$ be a L\'evy process and suppose that the
 L\'evy measure $\nu$ of $Y_1$
 satisfies \eqref{eq:levymcond0} with $p>0$.
 Let $Z$ be a semimartingale such that 
 $\inf_{t \in (0,1]}\Delta Z_t > -1$ a.s., 
 $[Z,Y] = 0$ a.s., and for some $\delta > 0$
 \begin{align}\label{eq:momcond}
   \E\sup_{t\in [0,1]} \calE(Z)_{t}^{-\alpha-\delta} <\infty. 
 \end{align}
 Then the solution $X^{\vep}$
 to \eqref{eq:rr2} satisfies
 \begin{align*}
   \lim_{\vep\to 0} 
   \frac{\Prob(\inf_{t\in [0,1]}X^{\vep}_t<0)}{\nu(-\infty,-\vep^{-1})}
    = x^{-\alpha}\int_{0}^1\E \calE(Z)_{t}^{-\alpha} dt.
  \end{align*}
\end{prop}
Note that the moment condition \eqref{eq:momcond} only concerns the
behavior of $\calE(Z)$ near $0$. This conditions implies
that the probability that the unperturbed system $X^0$ is close to $0$
is sufficiently small. 
If $Z$ is a L\'evy process satisfying $\inf_{t \in (0,1]}\Delta Z_t > -1$ 
a.s.,
then whether \eqref{eq:momcond} holds or not depends only on the decay of 
the L\'evy measure of $Z_1$ near $-1$. In this case the following is a
more easily checked sufficient condition.

\begin{prop}\label{prop:momcond}
  Let $Z$ be a L\'evy process and let $\eta$ be the L\'evy measure of $Z_1$.
  If $\eta(-\infty,-1]=0$ and 
  $\int_{-1}^{-a} (1+z)^{-\alpha - \delta}\eta(dz) < \infty$
  for some $a \in (0,1)$, then \eqref{eq:momcond} holds.
\end{prop}
Since this is a special case of Proposition \ref{prop:suffcondlevy}
below we omit the proof. 

\begin{rem}
  Note that for the L\'evy process $Z$ in Proposition \ref{prop:momcond},
  $\eta(-\infty,-1]=0$ implies that $\inf_{t \in (0,1]}\Delta Z_t > -1$ a.s.
\end{rem}

The moment condition on the L\'evy measure $\eta$ is such that the
distribution of the jumps of the L\'evy process $Z$ can be regularly
varying at $-1$ as long as the index of regular variation is strictly
less than $-\alpha$. That is, the risky asset may, for instance, have
heavy-tailed negative returns as long as the 
tail is not too heavy compared to that of the L\'evy measure $\nu$. 

If $Z$ is a L\'evy process, then the constant 
$\int_{0}^1\E\calE(Z)_{t}^{-\alpha}dt$ appearing in Proposition \ref{thm:old}
can be explicitly computed. 

\begin{prop}\label{prop:asymconstant}
  Let $Z$ be a L\'evy process on the form
  $Z_t = r t + \sigma B_t + J_t$,
  where $r \in \R$, $\sigma \geq 0$, $B$ is a standard Brownian motion
  and $J$ is a compound Poisson process independent of $B$.
  If $J=0$, then
  \begin{align*}
    \int_{0}^1\E \calE(Z)_{t}^{-\alpha} dt
    = C(\alpha,r,\sigma)
    = \left\{\begin{array}{ll}
        \frac{\exp\left\{\frac{\sigma^2}{2}(\alpha^2+\alpha)-\alpha
          r\right\}-1}{\frac{\sigma^2}{2}(\alpha^2+\alpha)-\alpha r}
        & \text{ if } \alpha \neq 2r/\sigma^2 -1,\\
        1 & \text{ if } \alpha = 2r/\sigma^2 -1.
      \end{array}\right.    
  \end{align*}
  If $J\neq 0$ and the L\'evy measure $\eta$ of $J_1$ satisfies
  $\eta(-\infty,-1]=0$ and 
  $\int_{-1}^{-a} (1+z)^{-\alpha}\eta(dz) < \infty$
  for some $a \in (0,1)$, then
  \begin{align*}
    \int_{0}^1\E \calE(Z)_{t}^{-\alpha} dt
    = C(\alpha,r,\sigma)\frac{1}{\lambda}(e^\lambda -1)
  \end{align*}
  with
  $\lambda
  =\exp\{\eta(\R)(\exp\{\eta(\R)^{-1}\int (1+x)^{-\alpha}\eta(dx)\}-1)\}$.
\end{prop}
The proof is given in Section \ref{sec:proofs}.
\begin{rem}
  The expectation $\E \calE(Z)_{t}^{-\alpha}$ can be computed explicitly also
  in the case when the L\'evy process $J$ is not necessarily 
  a compound Poisson process by combining Theorem 25.17 in \cite{S99} and 
  Lemma 2.7 in \cite{KS02}.
  However, this results in a more complicated expression.
\end{rem}

Proposition \ref{thm:old} is sufficient for determining the asymptotic decay of
finite time ruin probabilities in quite general models. Not surprisingly, the
result can be shown to hold without any assumption about the decay of
the right tail of the L\'evy measure $\nu$. 
Indeed, it is only the
negative jumps of $\vep Y$ that can cause the process $X^\vep$ 
to take negative values. 
What is more important is that the result is very robust to changes in
the semimartingale $Z$. Next we explore this robustness in detail.

Throughout the rest of this paper we weaken the assumption
\eqref{eq:levymcond0} and only assume that the L\'evy measure $\nu$
of $Y_1$ has a regularly varying left tail. That is, for some $\alpha > 0$,
\begin{align}\label{eq:levymcond}
  \lim_{u\to\infty}
  \frac{\nu(-\infty,-\lambda u)}{\nu(-\infty,-u)}=\lambda^{-\alpha},
  \quad \lambda > 0.
\end{align}
In particular, the right tail of $\nu$ is allowed to decay arbitrarily slowly.
Proposition \ref{thm:old} can be extended to hold uniformly over a family
of semimartingales in the sense of the following theorem.

\begin{thm}\label{thm:hittingprob}
  Let $Y$ be a L\'evy process and suppose that the
  L\'evy measure $\nu$ of 
  $Y_1$ satisfies \eqref{eq:levymcond}.
  Let $\Gamma$ be any non-empty family of semimartingales $Z$  
  such that $[Z,Y] = 0$ a.s.~and 
  $\inf_{t \in (0,1]}\Delta Z_t > -1$ a.s.~for every $Z \in \Gamma$,
  and such that  
  \begin{align}\label{eq:momcondtype1}
    \sup_{Z\in\Gamma}\E\sup_{t\in [0,1]}
    \calE(Z)_{t}^{-\alpha-\delta} < \infty
    \quad \text{for some } \delta > 0.
  \end{align}
  Then the solutions $X^{\vep}=X^{\vep,Z}$, for $Z \in \Gamma$, 
  to \eqref{eq:rr2} satisfy
  \begin{align*}
    \lim_{\vep\to 0}\sup_{Z\in\Gamma}\Big|
    \frac{\Prob(\inf_{t\in [0,1]}X^{\vep,Z}_t<0)}{\nu(-\infty,-\vep^{-1})}
    -x^{-\alpha}\int_{0}^1\E \calE(Z)_{t}^{-\alpha} dt\Big|=0
  \end{align*}
  and
  \begin{align*}
    \lim_{\vep\to 0}\inf_{Z \in \Gamma}
    \frac{\Prob(\inf_{t\in [0,1]}X^{\vep,Z}_t < 0)}{\nu(-\infty,-\vep^{-1})}
    = \inf_{Z \in \Gamma}x^{-\alpha}\int_{0}^1\E \calE(Z)_{t}^{-\alpha} dt.
  \end{align*}
  In particular, if there exists $Z^* \in \Gamma$ such that 
  \begin{align*}
    \int_{0}^1\E \calE(Z^*)_{t}^{-\alpha} dt
    = \inf_{Z \in \Gamma}\int_{0}^1\E \calE(Z)_{t}^{-\alpha} dt,
  \end{align*}
  then
  \begin{align*}
    \lim_{\vep\to 0}\inf_{Z \in \Gamma}
    \frac{\Prob(\inf_{t\in [0,1]}X^{\vep,Z}_t < 0)}{\nu(-\infty,-\vep^{-1})}
    = \lim_{\vep\to 0}
    \frac{\Prob(\inf_{t\in [0,1]}X^{\vep, Z^*}_t < 0)}
	 {\nu(-\infty,-\vep^{-1})}.
  \end{align*}
\end{thm}
The proof is given in Section \ref{sec:proofs}.

\section{Asymptotic decay of finite time ruin probabilities}
\label{sec:ruin}

Consider an insurance company whose cumulative premiums minus claims
are modeled by a L\'evy process $\vep Y$. The
L\'evy measure $\nu$ of $Y_1$ is assumed to satisfy
\eqref{eq:levymcond}. That is, the left tail of $\nu$ is regularly
varying. 
Suppose that the insurance company has the
opportunity to use a dynamic investment strategy. Assume that there are $n$ 
risky assets whose spot prices $S^k_t$ form strictly positive semimartingales
and a bank account that gives non-negative instantaneous interest
rate $r_t$, where 
$r = \{r_t\}_{t \in [0,1]}$ 
is a \cadlag\ adapted stochastic process.
Let $\pi = \{(\pi^0_t,\dots,\pi^n_t)\}_{t\in [0,1]}$ be a \caglad\ predictable 
stochastic process, 
where $\pi^k_t$ denotes the fraction of the risk reserve invested in the $k$th
risky asset and $\pi^0_t=1-\pi^1_t-\dots -\pi^n_t$ is the fraction invested 
in the bank account, at time $t$. 
With this notation the evolution of the risk reserve follows
the stochastic integral equation
\begin{align}\label{eq:x}
  X^{\vep, \pi}_t = x + \int_{0+}^t \pi^0_{s} X^{\vep,\pi}_{s-}r_{s-} ds +
  \sum_{k=1}^n\int_{0+}^t \pi^k_{s} 
  X^{\vep,\pi}_{s-}\frac{dS^k_{s}}{S^k_{s-}} + \vep Y_t, 
  \quad t \in [0,1].
\end{align}
Since $S^k$ is a strictly positive semimartingale, $S^k_t=\calE(U^k)_t$, 
where $U^k_t$ is the semimartingale given by 
$U^k_t = S^k_0 + \int_{0+}^t (S^k_{s-})^{-1}dS^k_s$ 
(see Lemma 2.2 in \cite{KS02}).
Hence, $X^\vep$ is of the form \eqref{eq:rr2} where 
the semimartingale $Z^\pi$ is
given by
\begin{align}\label{eq:z}
  Z^{\pi}_t = \int_{0+}^t \pi^0_{s} r_{s-} ds + \sum_{k=1}^n\int_{0+}^t
  \pi^k_{s}\frac{dS^k_s}{S^k_{s-}},
  \quad t \in [0,1].
\end{align} 
Note also that if $[S^k,Y]=0$ for all $k$, then $[Z^{\pi},Y]=0$. 


An investment strategy $\pi$ will be called {\it optimal} if it
minimizes the ruin 
probability within a reasonably large class of strategies. That is, 
$\pi^*(\vep)$ is optimal if
\begin{align*}
  P\Big(\inf_{t \in [0,1]}X^{\vep,\pi^*(\vep)}_t < 0\Big) \leq
  P\Big(\inf_{t \in [0,1]}X^{\vep, \pi}_t < 0\Big), 
\end{align*}
for every strategy $\pi$ in the class. It is generally difficult to
find optimal strategies, even in relatively simple models, and it typically
involves solving a partial differential equation. 
An easier problem is to look for, what we will call, an
{\it asymptotically optimal strategy}. That is, a strategy $\pi^*_{as}$
that minimizes $\int_{0}^1 E \calE(Z^\pi)_{t}^{-\alpha} dt$.
Using Theorem \ref{thm:hittingprob} we find that, for small $\vep$,
the ruin probability for the optimal strategy $\pi^*(\vep)$ is not
much smaller than for an asymptotically optimal investment strategy
$\pi^*_{as}$. More precisely, the asymptotic decay under an optimal strategy
is the same as under an asymptotically optimal strategy. 
We summarize the findings of this section in the following corollary to
Theorem \ref{thm:hittingprob}.

\begin{cor}\label{thm:asoptimal2}
  Let $Y$ be a L\'evy process and suppose that the
  L\'evy measure $\nu$ of $Y_1$ satisfies \eqref{eq:levymcond}.
  Let $X^{\vep,\pi}$ be the solution to \eqref{eq:x}, where each  
  strictly positive semimartingale $S^k$ satisfies $[S^k,Y] = 0$ a.s. and
  $\pi$ belongs to a non-empty family $\Pi$ of \caglad\ predictable processes. 
  Suppose that $\Gamma = \{Z^{\pi};\pi\in\Pi\}$, where $Z^{\pi}$ is given by
  \eqref{eq:z}, satisfies the conditions of
  Theorem \ref{thm:hittingprob}. Then
  \begin{align*}
    \lim_{\vep\to 0}\sup_{\pi\in\Pi}\Big|
    \frac{\Prob(\inf_{t\in [0,1]}X^{\vep,\pi}_t < 0)} 
	 {\nu(-\infty,-\vep^{-1})}       
	 - x^{-\alpha}
	 \int_{0}^1 \E \calE(Z^{\pi})_{t}^{-\alpha} dt\Big|=0
  \end{align*}
  and
  \begin{align*}
    \lim_{\vep\to 0} \inf_{\pi\in\Pi}
    \frac{\Prob(\inf_{t\in [0,1]}X^{\vep, \pi}_t < 0)} 
	 {\nu(-\infty,-\vep^{-1})}       
    &= \inf_{\pi\in\Pi}\lim_{\vep\to 0}
    \frac{\Prob(\inf_{t\in [0,1]}X^{\vep, \pi}_t < 0)} 
	 {\nu(-\infty,-\vep^{-1})}\\
	 &= \inf_{\pi\in\Pi}x^{-\alpha}
	 \int_{0}^1 \E \calE(Z^{\pi})_{t}^{-\alpha} dt.
  \end{align*}
\end{cor}

\begin{rem}
The conditions of Theorem \ref{thm:hittingprob} that $\Gamma$ needs to satisfy
have natural interpretations. First, it is assumed that 
that $\inf_{t \in [0,1]} \Delta Z^\pi_t > -1$
a.s.~for all $\pi\in \Pi$. This means
that the company cannot be ruined simply by investing in the risky
assets. At least one insurance claim is necessary for the risk reserve
to become negative. The second condition is that 
\begin{align*}
  \sup_{\pi \in \Pi} 
  E \sup_{t \in [0,1]} \calE(Z^\pi)_t^{-\alpha-\delta} < \infty
  \quad \text{for some } \delta > 0.
\end{align*}
This condition says that it is sufficiently unlikely
that the company is near ruin due to unsuccessful investments only. 
\end{rem}

In the setting of (the first statement of) Corollary
\ref{thm:asoptimal2} it would be natural 
to consider strategies $\pi$ for which $\pi_t$ is some function of 
the reserve process, the interest rate and asset prices, and the 
premiums-minus-claims process up to (but not including) time $t$.
In this case we might take 
\begin{align*}
  \pi_t=\pi^\vep_t
  =f(X^{\vep}_{t-},r_{t-},S^1_{t-},\dots,S^n_{t-},\vep Y_{t-})
\end{align*}
for some function $f$ and set $\Pi=\{\pi^{\vep}; \vep \geq 0\}$.
In a given application one would choose a suitable small $\vep > 0$
and use the approximation
\begin{align*}
  \Prob\Big(\inf_{t\in [0,1]}X^{\vep,\pi^{\vep}}_t < 0\Big) 
  \approx\nu(-\infty,-\vep^{-1})x^{-\alpha}
  \int_{0}^1 \E \calE(Z^{\pi^{\vep}})_{t}^{-\alpha} dt
\end{align*} 
to estimate the ruin probability.

We will now present two specific models or sets of assumptions for which 
the conditions of Theorem \ref{thm:hittingprob} hold and hence the conclusions 
of Corollary \ref{thm:asoptimal2} hold. We note that whether the moment
condition \eqref{eq:momcondtype1} holds depends both on 
the model for the risky assets $S^k$ and the set of
investment strategies $\Pi$. 

Consider first the case where the 
dynamics for the risky assets are
given by $S^k = \calE(U^k)$ for L\'evy processes $U^k$ and where
the investment strategies rule out short-selling.  

\begin{prop}\label{prop:suffcondlevy}
  Let $Y$ be a L\'evy process and suppose that the
  L\'evy measure $\nu$ of $Y_1$ satisfies \eqref{eq:levymcond} for some 
  $\alpha > 0$.
  Let $X^{\vep,\pi}$ be the solution to \eqref{eq:x}, where 
  $r_t \geq 0$ and, for each $k$,
  $S^k$ satisfies $[S^k,Y] = 0$ a.s.~and
  is given by $S^k = \calE(U^k)$ for a L\'evy process
  $U^k$ for which the L\'evy measure $\eta^k$ of $U^k_1$ satisfies
  $\eta^k(-\infty,-1]=0$ and
  $\int_{-1}^{-a}(1+u)^{-n\alpha-\delta}\eta^k(du)<\infty$ for some 
  $a\in (0,1)$ and $\delta > 0$.
  Suppose that $\pi$ belongs to a family $\Pi$ of \caglad\ predictable
  processes such that $\pi^k_t \in [0,1]$ for all 
  $k\in\{0,1,\dots,n\}$. Then
  the moment condition \eqref{eq:momcondtype1} of Theorem \ref{thm:hittingprob}
  holds and also the conclusion of Corollary \ref{thm:asoptimal2}, i.e.
  \begin{align*}
    \lim_{\vep\to 0}\sup_{\pi\in\Pi}\Big|
    \frac{\Prob(\inf_{t\in [0,1]}X^{\vep, \pi}_t < 0)} 
	 {\nu(-\infty,-\vep^{-1})}  
	 - x^{-\alpha}
	 \int_{0}^1 \E \calE(Z^{\pi})_{t}^{-\alpha} dt\Big|=0
  \end{align*}
  and 
  \begin{align*}
    \lim_{\vep\to 0} \inf_{\pi\in\Pi}
    \frac{\Prob(\inf_{t\in [0,1]}X^{\vep, \pi}_t < 0)} 
	 {\nu(-\infty,-\vep^{-1})}       
    &= \inf_{\pi\in\Pi}\lim_{\vep\to 0}
    \frac{\Prob(\inf_{t\in [0,1]}X^{\vep, \pi}_t < 0)} 
	 {\nu(-\infty,-\vep^{-1})}\\
	 &= \inf_{\pi\in\Pi}x^{-\alpha}
	 \int_{0}^1 \E \calE(Z^{\pi})_{t}^{-\alpha} dt,
  \end{align*}
  where $Z^{\pi}$ is given by \eqref{eq:z}.
\end{prop}

In the second case the dynamics for the risky assets are diffusions. 
In order to obtain explicit results such as asymptotically optimal strategies
we only consider the case $n=1$ and set $\pi^1_t=\pi_t$ and
$\pi^0_t=1-\pi_t$ for a \caglad\ predictable process $\pi$.

Suppose that the asset price process $S$ is a solution to a stochastic 
integral equation of the form
\begin{align}\label{eq:s}
  S_t = S_0 + \int_{0+}^t \mu_{s-} S_{s-} ds 
  + \int_{0+}^t \sigma_{s-} S_{s-} dB_s,
  \quad t \in [0,1],
\end{align}
where $\mu$ and $\sigma$ 
are \cadlag\ adapted processes
with $\inf_{t \in [0,1]} \sigma_t > 0$ a.s., and 
$B$ is a Brownian motion. 

A sufficient condition for the moment condition \eqref{eq:momcondtype1} in 
Theorem \ref{thm:hittingprob} to hold is given next. 

\begin{prop}\label{thm:asoptimalconst}
  Take $\alpha > 0$. 
  Suppose that the evolution of the risk reserve follows \eqref{eq:x},  
  where $S$ is given by \eqref{eq:s} and 
  $\pi$ belongs to a family $\Pi$ of 
  \caglad\ predictable processes for which for all $p>0$ and some 
  $\gamma > \alpha$
  \begin{align*}
    & \E\exp\left\{(2\gamma^2+\gamma)\int_{0}^1\pi_{t}^2\sigma_{t}^2dt
    \right\} < \infty,\\
    & \E \sup_{t\in [0,1]}
    \exp\left\{p \int_{0}^t[r_{s}+\pi_{s}(\mu_{s}-r_{s})]ds\right\} < \infty.
  \end{align*}
  If $Z^\pi$ is given by \eqref{eq:z}, then
  there exists a $\delta>0$ such that
  \begin{align*}
    \sup_{\pi\in\Pi}\E \sup_{t\in [0,1]}
    \calE(Z^{\pi})_{t}^{-\alpha-\delta} < \infty.
  \end{align*}
\end{prop}

\begin{exmp}
  Let the dynamics of $V_t=\sigma_t^2$ be given by the 
  Cox-Ingersoll-Ross (CIR) model
  \begin{align*}
    V_t = V_0 + \kappa \int_0^t (\theta-V_s)ds 
    + \delta \int_0^t \sqrt{V_s}dW_s,
  \end{align*}
  where $\kappa,\theta,\delta$ are positive constants and $W$ is standard
  Brownian motion. 
  Corollaries 3.2 and 3.3 in \cite{AP07} give necessary and sufficient 
  conditions for the integrated squared volatility process to have finite 
  exponential moments:
  \begin{align}\label{eq:finexpmom}
    \E \exp\left\{u\int_0^t V_s ds\right\} < \infty, \quad u > 0.
  \end{align}
  If $u \leq \kappa^2/(2\delta^2)$, then  \eqref{eq:finexpmom} holds
  for all $t > 0$. If $u > \kappa^2/(2\delta^2)$, then
  \eqref{eq:finexpmom} holds for all $t < t^*$, where
  $t^* = 2\gamma^{-1}(\pi+\arctan(-\gamma/\kappa))$ with 
  $\gamma = \sqrt{2\delta^2 u - \kappa^2}$.
\end{exmp}

When the risky asset is modeled by \eqref{eq:s} it is possible to
find the asymptotically optimal strategy explicitly. 

\begin{prop}\label{thm:asoptimal}
  Take $\alpha > 0$. 
  Suppose that the evolution of the risk reserve follows \eqref{eq:x},  
  where $S$ is given by \eqref{eq:s} and 
  $\pi$ belongs to the family $\Pi$ of 
  \caglad\ predictable processes for which 
  \begin{align}\label{eq:comcond}
    \E\exp\left\{\frac{\alpha^2}{2}\int_{0}^1\pi_{t}^2\sigma_{t}^2dt
    \right\} < \infty.
  \end{align}
  If $Z^\pi$ is given by \eqref{eq:z} and if  
  $\pi^*$ is given by
  $\pi^*_t = \frac{\mu_{t}-r_{t}}{(1+\alpha)\sigma^2_{t}}$
  and satisfies \eqref{eq:comcond},
  then 
  \begin{align*}
    \int_{0}^1 \E \calE(Z^{\pi^*})_{t}^{-\alpha} dt \leq 
    \int_{0}^1 \E \calE(Z^\pi)_{t}^{-\alpha} dt
  \end{align*}
  for every $\pi \in \Pi$. 
\end{prop}

\begin{rem}
Note that the asymptotically optimal investment strategy looks just
like the solution to the Merton problem (see e.g.~\cite{FS06} p.~169) 
with HARA 
utility. This comes from the fact that here minimizing
$\int_{0}^1 \E \calE(Z^\pi)_{t}^{-\alpha} dt$ 
is equivalent to minimizing $\E\calE(Z^\pi)_t^{-\alpha}$ 
which is very similar to maximizing $\E\calE(Z^\pi)_t^{\alpha}$ 
as is done in the Merton problem. 
\end{rem}

\begin{exmp}\label{ex:gbm2} 
  Suppose $r$ and $\mu$ and $\sigma$ are constants, i.e.~the
  spot price process $S$ of the risky asset is a geometric Brownian motion. 
  Then, the
  asymptotically optimal strategy $\pi^*$ is given by 
  $\pi_t^*=\frac{\mu-r}{(1+\alpha)\sigma^2}$ 
  and the asymptotic decay of the 
  finite time ruin probability is 
   \begin{align*}
    \lim_{\vep\to 0}&\frac{\Prob(\inf_{t\in [0,1]}X^{\vep,\pi^*}_t < 0)}
    {\nu(-\infty,-\vep^{-1})}
    = x^{-\alpha}\int_{0}^1 \E \calE(Z^{\pi^*})_{t}^{-\alpha}dt \\
    &= x^{-\alpha}\int_0^1 \E \exp\left\{-\alpha(1-\pi^*)rt - \alpha
    \pi^*\mu t + \alpha\frac{(\pi^*)^2 \sigma^2}{2}t - \alpha \pi^*
    \sigma B_t\right\}dt \\
    &= x^{-\alpha}\int_0^1 \exp\left\{
  -\alpha(1-\pi^*)rt -\alpha \pi^*\mu t+
  (1+\alpha)\frac{\alpha}{2}(\pi^*)^2\sigma^2 t\right\} dt \\
   &= x^{-\alpha}\int_0^1 \exp\left\{\left(-\alpha r - \frac{\alpha
     (\mu-r)^2}{2(1+\alpha) \sigma^2}\right)t\right\} dt \\
   &= x^{-\alpha} \frac{1-\exp\{-\alpha r - \frac{\alpha
     (\mu-r)^2}{2(1+\alpha) \sigma^2}\}}{\alpha r + \frac{\alpha
     (\mu-r)^2}{2(1+\alpha) \sigma^2}}.
  \end{align*}
  This may be compared to the strategy $\pi = 0$ with no investment in
  the risky asset. Proposition \ref{prop:asymconstant} yields
  \begin{align*}
    \lim_{\vep\to 0}\frac{\Prob(\inf_{t\in [0,1]}X^{\vep,0}_t < 0)}
    {\nu(-\infty,-\vep^{-1})}
    = x^{-\alpha}\int_{0}^1 \E \calE(Z^{0})_{t}^{-\alpha}dt
    = x^{-\alpha} \frac{1-e^{-\alpha r}}{\alpha r}.
  \end{align*}
  Note that the reduction of the asymptotic decay of the ruin
  probability using the asymptotically optimal strategy compared to
  no investment depend crucially on the 
  (Sharpe) ratio $\gamma = (\mu - r)/\sigma$. 
  If the constant 
  \begin{align*}
    R =  \lim_{\vep\to 0}&\frac{\Prob(\inf_{t\in
        [0,1]}X^{\vep,\pi^*}_t < 0)}{\Prob(\inf_{t\in
        [0,1]}X^{\vep,0}_t < 0)} 
    = \frac{1-e^{-\alpha r}e^{-\alpha\gamma^2/2(1+\alpha)}}{1-e^{-\alpha r}}
    \frac{\alpha r}{\alpha r + \alpha\gamma^2/2(1+\alpha)}
  \end{align*}
  is studied for reasonable parameter choices, $(r,\alpha)=(0.05,2)$ say,
  then one finds that 
  it is necessary to have the opportunity to invest in a
  very attractive risky asset, $\gamma > 1$ say, to have any significant
  reduction of the ruin probability.  
\end{exmp}

As mentioned in the introduction, Example \ref{ex:gbm2} above is closely 
related to the studies in
\cite{HP00,GG02,S05} of the infinite horizon case with $r=0$. 
Translating the results to our notation the authors obtain the
following limit as $\vep \to 0$ of the optimal strategy $\pi^*(\vep)$:
\begin{align*} 
\lim_{\vep \to 0} \pi^*(\vep) =
\frac{\mu}{(1+\alpha)\sigma^2}.
\end{align*}
This coincides with the asymptotically optimal strategy calculated
above. 


\section{Proofs and auxiliary results}
\label{sec:proofs}
\begin{lem}\label{lem:sol}
  The stochastic integral equation \eqref{eq:rr2} has a unique solution 
  which is given by \eqref{eq:rr23}. 
\end{lem}

\begin{proof}[Proof of Lemma \ref{lem:sol}]
  First some notation. 
  Let $A_t = Z_t - \frac{1}{2}[Z,Z]^c_t$, 
  $B_t = \prod_{s \in (0,t]}(1+\Delta Z_s) e^{-\Delta Z_s}$, 
  $C_t = x+\vep \int_{0+}^t\frac{dY_s}{\calE(Z)_{s-}}$, 
  and $X^\vep_t = e^{A_t}B_tC_t$. 
  Then $[A,B]^c_t = [B,B]^c_t = [B,C]^c_t = 0$, and $[A,A]^c_t =[Z,Z]^c_t$. 
  By It\^o's formula (see \cite{P04} Theorem 33)
  \begin{align*}
  X^\vep_t - x &= \int_{0+}^t   X^\vep_{s-} dA_s + \int_{0+}^t
  e^{A_{s-}}C_{s-} dB_s + \int_{0+}^t e^{A_{s-}}B_{s-}dC_s 
  \\ &\quad +
  \frac{1}{2} \int_{0+}^t   X^\vep_{s-}d[A,A]^c_s + 
  \int_{0+}^t e^{A_{s-}}B_{s-} d[A,C]^c_s\\ 
  &\quad + \sum_{s \in (0,t]}\Big(
  X^\vep_{s} -   X^\vep_{s-} - X^\vep_{s-}\Delta A_s -
  e^{A_{s-}}C_{s-}\Delta B_s - e^{A_{s-}}B_{s-}\Delta C_s\Big)\\
  &=  \int_{0+}^t X^\vep_{s-} dZ_s + \sum_{s\in (0,t]}
  e^{A_{s-}}C_{s-} \Delta B_s + \vep \int_{0+}^t dY_s +
  \int_{0+}^t e^{A_{s-}}B_{s-} d[A,C]^c_s\\ 
  &\quad + \sum_{s\in (0,t]}\Big(
  X^\vep_{s} - X^\vep_{s-} - X^\vep_{s-}\Delta A_s -
  e^{A_{s-}}C_{s-}\Delta B_s - e^{A_{s-}}B_{s-}\Delta C_s\Big)\\
  &=  \int_{0+}^t X^\vep_{s-} dZ_s + \vep Y_t +
  \int_{0+}^t e^{A_{s-}}B_{s-} d[A,C]^c_s\\ 
  &\quad + \sum_{s\in (0,t]}\Big(
  e^{A_{s-}}B_{s-}(1+\Delta Z_s)(C_s-C_{s-})
  - e^{A_{s-}}B_{s-}\Delta C_s\Big)\\
  &= \int_{0+}^t   X^\vep_{s-} dZ_s + \vep Y_t +
  \int_{0+}^t e^{A_{s-}}B_{s-} d[A,C]^c_s\\ 
  &\quad + \sum_{s\in (0,t]}\Big(
  e^{A_{s-}}B_{s-}(1+\Delta Z_s)\Delta C_s 
  - e^{A_{s-}}B_{s-}\Delta C_s\Big)\\
  &=  \int_{0+}^t X^\vep_{s-} dZ_s + \vep Y_t +
  \int_{0+}^t e^{A_{s-}}B_{s-} d[A,C]^c_s + \sum_{s\in (0,t]}
  e^{A_{s-}}B_{s-}\Delta Z_s\Delta C_s\\
  &=  \int_{0+}^t   X^\vep_{s-} dZ_s + \vep Y_t +
  \int_{0+}^t e^{A_{s-}}B_{s-} d[A,C]_s\\ 
  &= \int_{0+}^t X^\vep_{s-} dZ_s + \vep Y_t + \vep [Z,Y]_t\\
  &=  \int_{0+}^t X^\vep_{s-} dZ_s + \vep Y_t.
  \end{align*}
\end{proof}

\begin{proof}[Proof of Proposition \ref{prop:asymconstant}]
  First consider the case $J=0$.
  The constant $C(\alpha,r,\sigma)$ is computed as follows
  \begin{align*}
    \int_{0}^1\E \calE(Z)_{t}^{-\alpha} dt
    &=\int_0^1 \E \exp\{-\alpha((r-\sigma^2/2)t+\sigma B_t)\}dt\\
    &= \int_0^1 \exp\{-\alpha(r-\sigma^2/2)t\}
    \E(\exp\{-\alpha\sigma B_t\})dt\\
    &= \int_0^1 \exp\{(\sigma^2(\alpha^2+\alpha)/2-\alpha r)t\}dt\\
    &= C(\alpha,r,\sigma).
  \end{align*}
  Now consider the case $J\neq 0$.
  Note that the Dolean-Dade exponential of a sum of two independent processes
  is the product of the two Dolean-Dade exponentials. 
  To complete the proof we just repeat the computations at the end of the
  proof of Proposition \ref{prop:momcond}. This gives
  \begin{align*}
    \calE(J)_{t}^{-\alpha}
    = \exp\{t\eta(\R)(\exp\{\eta(\R)^{-1}\int (1+x)^{-\alpha}\eta(dx)\}-1)\}.
  \end{align*}
\end{proof}

\begin{proof}[Proof of Theorem \ref{thm:hittingprob}]
  From \eqref{eq:events} it follows that,
  provided that the limit exists, 
  \begin{align*}
    &\lim_{\vep\to 0}\inf_{Z \in \Gamma}
    \frac{\Prob(\inf_{t\in [0,1]}X^{\vep,Z}_t < 0)}{\nu(-\infty,-\vep^{-1})}
    = \lim_{\vep\to 0}\inf_{Z \in \Gamma}
    \frac{\Prob(\inf_{t\in [0,1]}\int_{0+}^t \calE(Z)_{s-}^{-1}dY_s<-x/\vep)}
	 {\nu(-\infty,-\vep^{-1})}\\
    &\quad = x^{-\alpha}\lim_{\vep\to 0}\inf_{Z \in \Gamma}
    \frac{\Prob(\inf_{t\in [0,1]}\int_{0+}^t \calE(Z)_{s-}^{-1}dY_s<-x/\vep)}
	 {\nu(-\infty,-x/\vep)}.	 
  \end{align*}
  Applying Theorem \ref{unibreimanstint} below 
  completes the proof.
\end{proof}

\begin{thm}\label{unibreimanstint}
  Let $Y$ be a L\'evy process 
  such that the L\'evy measure $\nu$ of $Y_1$ satisfies
  \eqref{eq:levymcond} for some $\alpha>0$.
  Let $\bba$ be a family of \caglad\ predictable strictly positive processes
  satisfying $\sup_{A \in \bba}\E \sup_{t\in [0,1]}|A_t|^{\alpha+\vep}<\infty$
  for some $\vep > 0$.
  Then
  \begin{align*}
    &\textrm{(i)}\quad\lim_{x\to\infty}\inf_{A \in \bba}
    \frac{\Prob(\inf_{t\in [0,1]}\int_0^t A_sdY_s<-x)}{\nu(-\infty,-x)}
    = \inf_{A \in \bba}\int_0^1\E A_t^{\alpha}dt,\\
    &\textrm{(ii)}\quad\lim_{x\to\infty}\sup_{A \in \bba}\left|
    \frac{\Prob(\inf_{t\in [0,1]}\int_0^t A_sdY_s<-x)}{\nu(-\infty,-x)}
    - \int_0^1\E A_t^{\alpha}dt\right|
    = 0. 
  \end{align*}
\end{thm}

\begin{proof}
  We use the notation $(A \cdot Y)$ for the stochastic integral process
  given by $(A \cdot Y)_t = \int_0^t A_sdY_s$.
  We first show that (ii) implies (i):
  \begin{align*}
    &\limsup_{x\to\infty}\inf_{A \in \bba}
    \frac{\Prob(\inf_{t\in [0,1]}(A\cdot Y)_t<-x)}{\nu(-\infty,-x)}\\
    &\quad \leq
    \limsup_{x\to\infty}\inf_{A \in \bba}
    \left(\frac{\Prob(\inf_{t\in [0,1]}(A\cdot Y)_t<-x)}{\nu(-\infty,-x)}
    -\int_0^1\E A_t^{\alpha}dt+\int_0^1\E A_t^{\alpha}dt\right)\\
   &\quad\leq \limsup_{x\to\infty}\sup_{A \in \bba}
    \left|\frac{\Prob(\inf_{t\in [0,1]}(A\cdot Y)_t<-x)}{\nu(-\infty,-x)}
    - \int_0^1\E A_t^{\alpha}dt\right|
    + \inf_{A \in \bba}\int_0^1\E A_t^{\alpha}dt\\
    &\quad = \inf_{A \in \bba}\int_0^1\E A_t^{\alpha}dt,\\
    &\liminf_{x\to\infty}\inf_{A \in \bba}
    \frac{\Prob(\inf_{t\in [0,1]}(A\cdot Y)_t<-x)}{\nu(-\infty,-x)}\\
    &\quad =
    \liminf_{x\to\infty}\inf_{A \in \bba}
    \left(\frac{\Prob(\inf_{t\in [0,1]}(A\cdot Y)_t<-x)}{\nu(-\infty,-x)}
    -\int_0^1\E A_t^{\alpha}dt+\int_0^1\E A_t^{\alpha}dt\right)\\
    &\quad \geq 
    \liminf_{x\to\infty}\inf_{A \in \bba}
    \left(\frac{\Prob(\inf_{t\in [0,1]}(A\cdot Y)_t<-x)}{\nu(-\infty,-x)}
    -\int_0^1\E A_t^{\alpha}dt\right)+\inf_{A \in \bba}\int_0^1\E A_t^{\alpha}dt\\
    &\quad = \inf_{A \in \bba}\int_0^1\E A_t^{\alpha}dt,
  \end{align*}
  Hence, (ii) implies (i). 

  It remains to show (ii).
  We decompose $Y$ (the L\'evy-It\^o decomposition) 
  into a sum $Y = \widetilde{Y} + J$ of 
  independent L\'evy processes, where $\widetilde{Y}$ has jumps 
  whose norms are bounded by 1 and $J=Y-\widetilde{Y}$ 
  is a compound Poisson process with representation
  $J_t = \sum_{k=1}^{N_t}Z_k$. 
  Moreover, 
  we can decompose $J$ into a sum $J = J_x + (J-J_x)$ of 
  independent compound Poisson processes, where $J_x$ consists of the
  jumps $\Delta Y_t$ of $Y$ with $|\Delta Y_t| > x^{\beta}$ for 
  some $\beta \in (1/2,1)$.
  Let $M_x = \#\{t \in (0,1] ; |\Delta Y_t| > x^{\beta}\}$ and let
  $\tau_{x,1},\dots,\tau_{x,M_x}$ be the time points of these jumps.
  Then $J_x=\{J_x(t)\}_{t\in [0,1]}$ is given by
  $J_x = \sum_{k=1}^{M_x} Z^x_k I_{[\tau_{x,k},1]}$, where 
  $Z^x_k = \Delta Y_{\tau_{x,k}}$.
  Note that, for any $\delta > 0$,
  \begin{align*}
    \Prob(\inf_{t\in [0,1]}(A\cdot Y)_t < -x)
    &= \Prob(\inf_{t\in [0,1]}(A\cdot Y)_t < -x, \sup_{t\in [0,1]}|(A\cdot\widetilde{Y})_t|>\delta x)\\
    &\quad + \Prob(\inf_{t\in [0,1]}(A\cdot Y)_t < -x, \sup_{t\in [0,1]}|(A\cdot\widetilde{Y})_t|\leq\delta x)\\
    &\leq \Prob(\sup_{t\in [0,1]}|(A\cdot\widetilde{Y})_t|>\delta x)\\
    &\quad + \Prob(\inf_{t\in [0,1]}(A\cdot J)_t < -(1-\delta)x)
\end{align*}
and
\begin{align*}
    \Prob(\inf_{t\in [0,1]}(A\cdot Y)_t < -x)
    &= \Prob(\inf_{t\in [0,1]}(A\cdot Y)_t < -x, \sup_{t\in [0,1]}|(A\cdot\widetilde{Y})_t|>\delta x)\\
    &\quad + \Prob(\inf_{t\in [0,1]}(A\cdot Y)_t < -x, \sup_{t\in [0,1]}|(A\cdot\widetilde{Y})_t|\leq\delta x)\\
    &\geq \Prob(\inf_{t\in [0,1]}(A\cdot J)_t < -(1+\delta)x).
  \end{align*}
  Hence, in order to prove (ii) it is sufficient to prove that
  \begin{align}\label{eq:proof2}
    \lim_{x\to\infty}\sup_{A \in \bba}\left|
    \frac{\Prob(\inf_{t\in [0,1]}(A\cdot J)_t<-x)}{\nu(-\infty,-x)}
    - \int_0^1\E A_t^{\alpha}dt\right|
    = 0
  \end{align}
  and that
  \begin{align}\label{eq:proof3}
    \lim_{x\to\infty}\sup_{A \in \bba}
    \frac{\Prob(\inf_{t\in [0,1]}(A\cdot \widetilde{Y})_t<-x)}{\nu(-\infty,-x)}
    = 0.
  \end{align}
  Similarly, in order to prove \eqref{eq:proof2} it is sufficient to prove
  that
  \begin{align}\label{eq:proof4}
    \lim_{x\to\infty}\sup_{A \in \bba}\left|
    \frac{\Prob(\inf_{t\in [0,1]}(A\cdot J_x)_t<-x)}{\nu(-\infty,-x)}
    - \int_0^1\E A_t^{\alpha}dt\right|
    = 0
  \end{align}
  and that
  \begin{align}\label{eq:proof5}
    \lim_{x\to\infty}\sup_{A \in \bba}
    \frac{\Prob(\inf_{t\in [0,1]}(A\cdot (J-J_x))_t<-x)}{\nu(-\infty,-x)}
    = 0.
  \end{align} 
  However, \eqref{eq:proof3} follows from Lemma 5.5 in \cite{HL07} 
  (Lemma 5.5 in \cite{HL07} is proved without the supremum over $\bba$ 
  but the proof holds also for the present stronger statement).
  We now show \eqref{eq:proof5}.
  Decompose $J-J_x$ into the sum $J-J_x=(J-J_x)^++(J-J_x)^-$, where
  \begin{align*}
    (J-J_x)^+_t=\sum_{k=1}^{N_t}Z_kI_{[1,x^\beta]}(Z_k),
    \quad 
    (J-J_x)^-_t=\sum_{k=1}^{N_t}Z_kI_{[-x^\beta,-1]}(Z_k).
  \end{align*}
  Note that
  \begin{align}
    \sup_{A \in \bba}
    \frac{\Prob(\inf_{t\in [0,1]}(A\cdot (J-J_x))_t<-x)}{\nu(-\infty,-x)}
    &\leq
    \sup_{A \in \bba}
    \frac{\Prob(\inf_{t\in [0,1]}(A\cdot (J-J_x)^-)_t<-x)}{\nu(-\infty,-x)}
    \nonumber \\
    &= 
    \sup_{A \in \bba}
    \frac{\Prob(\sup_{t\in [0,1]}(A\cdot [-(J-J_x)^-])_t>x)}{\nu(-\infty,-x)}
    \label{eq:proof5b}
  \end{align}
  and that \eqref{eq:proof5b} $\to 0$ as $x\to\infty$ by 
  Lemma 5.3 and Remark 5.1 in \cite{HL07}
  (Lemma 5.3 in \cite{HL07} is proved without the supremum over $\bba$ 
  but the proof holds also for the present stronger statement).
  Hence, we have shown \eqref{eq:proof5}.

  It remains to prove \eqref{eq:proof4}.
  Let 
  \begin{align*}
    M_x^- = \#\{t \in (0,1] ; \Delta Y_t < -x^{\beta}\},
    \quad 
    (J_x^{-})_t = \sum_{k=1}^{N_t}Z_k I_{(-\infty,-x^{\beta})}(Z_k).
  \end{align*}
  Note that
  \begin{align*}
    \Prob(A_{\tau_{x,1}}Z^x_1 < -x, M_x = 1)
    &\leq \Prob(\inf_{t\in [0,1]}(A\cdot J_x)_t<-x)\\
    &\leq \Prob(\inf_{t\in [0,1]}(A\cdot J_x^{-})_t<-x)\\
    &= \Prob(\inf_{t\in [0,1]}(A\cdot J_x^{-})_t<-x, M_x^{-}=1)\\
    &\quad + \Prob(\inf_{t\in [0,1]}(A\cdot J_x^{-})_t<-x, M_x^{-}\geq 2)\\
    &= \Prob(A_{\tau_{x,1}}Z^x_1 < -x, M_x = 1)\\
    &\quad + \Prob(\inf_{t\in [0,1]}(A\cdot J_x^{-})_t<-x, M_x^{-}\geq 2)\\
    &\leq \Prob(A_{\tau_{x,1}}Z^x_1 < -x, M_x = 1) + \Prob(M_x^{-}\geq 2)
  \end{align*}
  and that $\lim_{x\to\infty}\Prob(M_x^{-}\geq 2)/\nu(-\infty,-x)=0$ 
  by Lemma 5.4 in \cite{HL07}.
  Applying Lemma \ref{keylemma} below shows \eqref{eq:proof4}
  and hence completes the proof.
\end{proof}

\begin{lem}\label{keylemma}
  With the notation above it holds that 
  \begin{align*}
    \lim_{x\to\infty}\sup_{A\in\bba}\left|
    \frac{\Prob(A_{\tau_{x,1}}Z^x_1 < -x, M_x = 1)}
	 {\nu(-\infty,-x)}-\int_0^1\E A_t^{\alpha}dt
    \right|=0.
  \end{align*}
\end{lem}
\begin{proof}
Let $\xi$ be the Poisson random measure with intensity measure 
$\text{Leb}\times\nu$, where $\text{Leb}$ is Lebesgue measure on $[0,1]$,
that determines the jumps of $Y$ and
note that 
\begin{align*}
  &\Prob(A_{\tau_{x,1}}Z^x_1 < -x, M_x = 1)\\ 
  &\quad= \int\Prob(\xi([0,1]\times (-\infty,-\max(x/y,x^\beta))=1))
  d\Prob(A_{\tau_{x,1}}\leq y)\\
  &\quad\leq \int \nu(-\infty,-x/y)e^{-\nu(-\infty,-x/y)}
  d\Prob(A_{\tau_{x,1}}\leq y)\\
  &\quad\leq \int \nu(-\infty,-x/y)
  d\Prob(A_{\tau_{x,1}}\leq y).
\end{align*}
Set
\begin{align*}
  \Delta(x,A) 
  &= \frac{\Prob(A_{\tau_{x,1}}Z^x_1 < -x,M_x=1)}
	{\nu(-\infty,-x)}-\int_0^1\E A_t^{\alpha}dt\\
  &= \underbrace{
    \frac{\Prob(A_{\tau_{x,1}}Z^x_1 < -x,M_x=1)}{\nu(-\infty,-x)}
    -\E A_{\tau_{x,1}}^{\alpha}}_{\Delta_1(x,A)}
  + \underbrace{
    \E A_{\tau_{x,1}}^{\alpha}-\int_0^1\E A_t^{\alpha}dt}_{\Delta_2(x,A)}
\end{align*}
We need to show two things:
\begin{align*}
  \text{(A):} \quad \lim_{x\to\infty}\sup_{A \in \bba}|\Delta_1(x,A)|=0
  \quad \text{and} \quad 
  \text{(B):} \quad \lim_{x\to\infty}\sup_{A \in \bba}|\Delta_2(x,A)|=0.
\end{align*}
(A): This is essentially a uniform version of what is often called
Breiman's result (see Lemma 2.2 in \cite{KM05}). 
Take an arbitrary $C>0$ and note that for $x$ sufficiently large
  \begin{align*}
    \Delta_1(x,A) 
    &= \int_{[0,C]}\left(\frac{\nu(-\infty,-x/y)e^{-\nu(-\infty,-x/y)}}
	  {\nu(-\infty,-x)}-y^{\alpha}
      \right)d\Prob(A_{\tau_{x,1}} \leq y)\\
      &\quad - \E A_{\tau_{x,1}}^{\alpha}I_{(C,\infty)}(A_{\tau_{x,1}})\\
      &\quad + \int_{(C,\infty)}
      \frac{\Prob(A_{\tau_{x,1}}Z^x_1 < -x, M_x = 1)}{\nu(-\infty,-x)} 
      d\Prob(A_{\tau_{x,1}} \leq y)\\
      &= \Delta_{11}(x,A) - \Delta_{12}(A) + \Delta_{13}(x,A). 
  \end{align*}
  Since $\sup_{A\in\bba}\E A_{\tau_{x,1}}^{\alpha}<\infty$,
  $\lim_{C\to\infty}\sup_{A\in\bba}\Delta_{12}(A)=0$.
  The uniform convergence theorem for regularly varying functions 
  (Theorem 1.5.2 in \cite{BGT87}) implies
  that for every $C > 0$
  \begin{align*}
    \sup_{A\in\bba}|\Delta_{11}(x,A)|
    &\leq 
    \sup_{A\in\bba}\int_{[0,C]}
    \left|\frac{\nu(-\infty,-x/y)e^{-\nu(-\infty,-x/y)}}
	  {\nu(-\infty,-x)}-y^{\alpha}\right|
    d\Prob(A_{\tau_{x,1}} \leq y)\\
    &\leq \sup_{y\in [0,C]}
    \left|\frac{\nu(-\infty,-x/y)e^{-\nu(-\infty,-x/y)}}
	  {\nu(-\infty,-x)}-y^{\alpha}\right|
    \to 0.
  \end{align*}
  The Potter bounds (Theorem 1.5.6 in \cite{BGT87}) says that 
  for any $A > 1$ and $\delta \in (0,\vep)$ there exists $x_0=x_0(A,\delta)$
  such that 
  \begin{align*}
    \frac{\nu(-\infty,-x/y)}{\nu(-\infty,-x)} \leq Ay^{\alpha+\delta}
  \end{align*}
  whenever $x,x/y \geq x_0$ and $y > C > 1$.
  Hence,
  \begin{align*}
    &\sup_{A\in\bba}|\Delta_{13}(x,A)|
    \leq  
    \int\frac{\nu(-\infty,-x/y)}{\nu(-\infty,-x)}
    d\Prob(A_{\tau_{x,1}} \leq y)\\
   &\quad\leq 
    A\sup_{A\in\bba}\int_{(C,x/x_0)}
    y^{\alpha+\delta}d\Prob(A_{\tau_{x,1}} \leq y)
    + A\sup_{A\in\bba}
    \frac{\Prob(A_{\tau_{x,1}} < -x/x_0)}{\nu(-\infty,-x)}\\
    &\to 0
  \end{align*}
  by first letting $n\to\infty$ and then $C\to\infty$,
  since $\sup_{A\in\bba}\E A_{\tau_{x,1}}^{\alpha+\vep}<\infty$.

  (B):
  We now show that $\lim_{x\to\infty}\sup_{A\in\bba}|\Delta_2(x,A)|=0$.
  Let $A_x = (-\infty,-x^{\beta}) \cup (x^{\beta},\infty)$
  and note that for all $t \in [0,1]$,
  \begin{align*}
    &\Prob(\tau_{x,1} \leq t)
    = \Prob(\xi([0,t] \times A_x) \geq 1 \mid \xi([0,1] \times A_x) \geq 1)
    = \frac{1-\exp\{-t\nu(A_x)\}}{1-\exp\{-\nu(A_x)\}},\\
    &\frac{d}{dt}\Prob(\tau_{x,1} \leq t)
    = \frac{\nu(A_x)\exp\{-t\nu(A_x)\}}{1-\exp\{-\nu(A_x)\}}.
  \end{align*}
  Since $\tau_{x,1}$ and $A_{\tau_{x,1}}$ are independent it holds that
  \begin{align*}
    \sup_{A\in\bba}|\E A_{\tau_{x,1}}^{\alpha}- \int_0^1\E A_t^{\alpha}dt|
    &= \sup_{A\in\bba}\left|\int_0^1\E A_t^{\alpha}
    \left(\frac{\nu(A_x)\exp\{-t\nu(A_x)\}}{1-\exp\{-\nu(A_x)\}}-1
    \right)dt\right|\\
    &\leq \sup_{A\in\bba}\E\sup_{t\in [0,1]}A_t^{\alpha}
    \int_0^1 \left|\frac{\nu(A_x)\exp\{-t\nu(A_x)\}}{1-\exp\{-\nu(A_x)\}}-1
    \right|dt\\
    &\to 0
  \end{align*}
  as $x\to\infty$ by the bounded convergence theorem.
  The proof is complete.
\end{proof}

\begin{proof}[Proof of Proposition \ref{prop:suffcondlevy}]
  We first prove the claim in the case $n=1$. Then we show that the 
  claim in the case of a general $n$ follows from the one-dimensional case
  by applying H\"older's inequality.

  Let $\alpha$ and $\delta$ be as in Proposition \ref{prop:momcond}. 
  According to the L\'evy-It\^o decomposition we can decompose $U$ into the sum
  of three independent L\'evy processes: $U=F+G+H$, where $F$ is a Gaussian
  process with drift, $G$ has zero mean and jumps satisfying 
  $|\Delta G_t|<\vep$ 
  for some small $\vep$, 
  and $H$ is a compound Poisson process. 
  Set $\pi:=\pi^1$ so that $\pi^0=1-\pi$. Then 
  \begin{align*}
    Z^\pi_t &=  \int_{0+}^t (1-\pi_{s})r_{s-} ds + \int_{0+}^t
    \pi_{s}dF^k_s + \int_{0+}^t
    \pi_{s}dG^k_s + \int_{0+}^t
    \pi_{s}dH^k_s\\  
    & =: \int_{0+}^t (1-\pi_{s}) r_{s-} ds + F^\pi_t + G^\pi_t + H^\pi_t.
  \end{align*}
  We note that 
  $\calE(Z^\pi)_t=e^{\int_{0+}^t(1-\pi_{s})r_{s-}ds}
  \calE(F^\pi)_t\calE(G^\pi)_t\calE(H^\pi)_t$ 
  and hence that
  \begin{align*}
    \E\sup_{t\in [0,1]} \calE(Z^\pi)_t^{-(\alpha+\delta/2)}
    \leq \E\Big(&\sup_{t\in [0,1]}
    e^{-(\alpha+\delta/2)\int_{0+}^t(1-\pi_{s-})r_{s-}ds}
    \sup_{t\in [0,1]} \calE(F^\pi)_t^{-(\alpha+\delta/2)} \\
    &\sup_{t\in [0,1]}\calE(G^\pi)_t^{-(\alpha+\delta/2)}\sup_{t\in [0,1]}
    \calE(H^\pi)_t^{-(\alpha+\delta/2)}\Big). 
  \end{align*}
  We note that $e^{-(\alpha+\delta/2)\int_{0+}^t(1-\pi_{s})r_{s-}ds}\leq 1$
  for all $t$ since $\pi_s \in [0,1]$ and $r_s\geq 0$ for all $s$.
  Using H\"olders inequality with $1< p < (\alpha+\delta)/(\alpha +
  \delta/2)$ and $1/p+1/q = 1$ the above expression is less than or equal to
  \begin{align*}
    \Big(\E \sup_{t\in [0,1]}
    \calE(F^\pi)_t^{-q(\alpha+\delta/2)}
    \calE(G^\pi)_t^{-q(\alpha+\delta/2)}\Big)^{1/q}\Big(\E\sup_{t\in
      [0,1]}\calE(H^\pi)_t^{-p(\alpha+\delta/2)}\Big)^{1/p}. 
  \end{align*}
  Using the Cauchy-Schwartz inequality and the fact that $\pi_{t}\in [0,1]$ 
  an upper bound for the above expression is 
  \begin{align*}
    \underbrace{\Big(\E \sup_{t\in 
        [0,1]}\calE(F^\pi)_t^{-2q(\alpha+\delta/2)}\Big)^{1/2q}}_{I} 
   \underbrace{\Big(\E \sup_{t\in
      [0,1]}\calE(G^\pi)_t^{-2q(\alpha+\delta/2)}\Big)^{1/2q}}_{II}
  \underbrace{\Big(\E\sup_{t\in
      [0,1]}\calE(H^\pi)_t^{-p(\alpha+\delta/2)}\Big)^{1/p}}_{III}.   
  \end{align*}
  The proof is complete when we have shown that each of these three factors
  exists finitely.
  We start with the first factor $I$ and show that, for any $\beta > 0$,
  \begin{align*}
    \E \sup_{t\in [0,1]}\calE(F^\pi)_t^{-\beta} < \infty. 
  \end{align*}
  Write $F_t = a t + \sigma B_t$ where $a \in \R$,
  $\sigma >0$ and $B$ is a Brownian motion.
  Then $\calE(F^\pi)$ is given by 
  $\calE(F^\pi)_t = \exp\{\int_0^t(a\pi_{s}-\sigma^2\pi_{s}^2/2)ds +
  \int_0^t \sigma \pi_{s}dB_s\}$ 
  and we have 
  \begin{align*}
    \E\sup_{t\in [0,1]}\calE(F^\pi)_{t}^{-\beta}
    & = \E\sup_{t\in [0,1]}\exp\{-\beta\sigma\int_0^t \pi_{s-}dB_s\} 
    \exp\{-\beta \int_0^t(a\pi_{s}-\sigma^2\pi_{s}^2/2)ds \}\\ 
    & \leq  \exp\{t\beta (\sigma^2/2- \min\{a,0\})\} 
    \E\sup_{t\in [0,1]}\exp\{-\beta\sigma\int_0^t \pi_{s}dB_s\}. 
  \end{align*}
  Set $M^I_t := - \sigma \int_0^t \pi_{s}dB_s$ 
  and note that for any $\lambda > 0$, $\lambda M^I$ is a
  continuous martingale and hence (see \cite{P04}, Theorem 39, p.~138)
  \begin{align*}
    E \exp\{\beta M^I_t\} \leq E \exp\{4 \beta^2 [M^I,M^I]_t\} 
    \leq E \exp\{4 \beta^2 \sigma^2 t\} < \infty. 
  \end{align*}
  Then Lemma \ref{below} below completes the proof of part $I$. 

  Next we consider $II$ and show that, for any $\beta > 0$,
  \begin{align*}
    \E\sup_{t\in [0,1]}\calE(G^\pi)_{t}^{-\beta}<\infty. 
  \end{align*}
  Denote by $\xi$ the Poisson random measure associated with the jumps of $G$ 
  such that  
  \begin{align*}
    G_t = \int_0^t \int_{\{|x| < \vep\}} x (\xi(ds,dx)-ds \eta(dx))
    =: \int_0^t \int_{\{|x| < \vep\}} x \tilde \xi(ds,dx).
  \end{align*}
  Then, by It\^o's formula (see also \cite{A04}, p.~248)
  \begin{align*}
    \calE(G^\pi)_t = \exp\Big \{&
    \int_0^t \int_{\{|x| < \vep\}} \log(1+\pi_{s}x) \tilde \xi(ds,dx)\\ 
    & + \int_0^t \int_{\{|x| < \vep\}} 
    \Big(\log(1+\pi_{s}x) - \pi_{s}x\Big) ds\eta(dx)\Big\}
  \end{align*}
  which gives
  \begin{align*}
    \calE(G^\pi)_t^{-\beta} = \exp\Big \{&
    -\beta \int_0^t \int_{\{|x| < \vep\}} 
    \log(1+\pi_{s}x) \tilde \xi(ds,dx)\\ 
    & -\beta \int_0^t \int_{\{|x| < \vep\}}
    \Big(\log(1+\pi_{s}x) - \pi_{s}x\Big) ds\eta(dx)\Big\}
  \end{align*}
  Set $M^{II}_t := -\int_0^t \int_{\{|x| < \vep\}} \log(1+\pi_{s}x)
  \tilde \xi(ds,dx)$ and note that (see e.g.~\cite{A04}, p.~209) that
  $M^{II}$ is a local martingale. 
  For $|y|<\vep$ and a constant $k=k(\vep)>0$ it holds that
  $|\log(1+y)-y| \leq k y^2$. Hence, since $|\pi_t|\leq 1$,
  \begin{align*}
    E \sup_{t\in [0,1]}\calE(G^\pi)_t^{-\beta} &\leq E\sup_{t\in [0,1]}
    e^{\beta M^{II}_t} \exp\Big\{\beta t \int_{\{|x| < \vep\}} k x^2
    \eta(dx)\Big\}\\
    & \leq K E\sup_{t\in [0,1]}
    e^{\beta M^{II}_t}. 
  \end{align*}
  Moreover, the quadratic variation of $M^{II}$ is given by
  (see \cite{A04}, p.~230)
  \begin{align*}
    [M^{II},M^{II}]_t=\int_0^t\int_{\{|x|<\vep\}}
    (\log(1+\pi_{s}x))^2\xi(ds,dx)
  \end{align*}
  and hence (see e.g.~Lemma 4.2.2, p.~197, in \cite{A04})
  \begin{align*}
    \E (M^{II}_t)^2 = \E [M^{II},M^{II}]_t 
    = \int_0^t \int_{\{|x| < \vep\}}\E (\log(1+\pi_{s}x))^2 \nu(dx)ds.
  \end{align*}
  This quantity is finite because $|\log(1+y)| \leq |y| + k y^2$ 
  for $|y|<\vep$ so it
  follows in particular that $M^{II}_t$ is a (square-integrable)
  martingale. 
  By Lemma \ref{below} below it is sufficient 
  to show $E e^{\beta M_t^{II}} < \infty$. 
  We introduce 
  \begin{align*}
    A_t = \frac{1}{2}\int_0^t \int_{\{|x| < \vep\}}
    -(1+\pi_{s}x)^{2\beta} + 1 + 2\beta \log(1+\pi_{s}x) \eta(dx)ds.
  \end{align*}
  By a Taylor expansion we get, for $|y|<\vep$ and a constant 
  $k=k(\vep)>0$,
  \begin{align*}
    |-(1+y)^{2\beta} + 1 + 2\beta\log(1+y)| \leq k y^2.
  \end{align*}
  This implies that $|A_t| < C t$ a.s.~for each $t$ and some
  constant $C > 0$. It follows by Cauchy-Schwartz inequality that 
  \begin{align*}
    E \exp\{\beta M_t^{II}\} \leq \Big(E\exp\{2\beta
    M^{II}_t-2A_t\}\Big)^{1/2} \Big(E \exp\{2A_t\}\Big)^{1/2}.
  \end{align*}
  We have constructed $A_t$ in such a way that
  $\exp\{2\beta M^{II}_t - 2A_t\}$ is a nonnegative local
  martingale starting at $1$; this follows from 
  Corollary 5.2.2, p.~253, in \cite{A04}.
  Hence $\exp\{2\beta M^{II}_t - 2A_t\}$ is also a supermartingale and 
  its expectation is bounded by $1$. Since $A_t$ is bounded we 
  finally arrive at $E \exp\{\beta M_t^{II}\} < \infty$. 
  This completes the proof of part $II$. 

  Finally we show that 
  $\E\sup_{t\in [0,1]}\calE(H^\pi)_{t}^{-\alpha-\delta}<\infty$.
  First we note that if $H^-$ consists of only the negative jumps of $H$,
  then
  \begin{align*}
    \E\sup_{t\in [0,1]}\calE(H^\pi)_{t}^{-\alpha-\delta} \leq
    \E\sup_{t\in [0,1]}\calE(H^{-,\pi})_{t}^{-\alpha-\delta}.
  \end{align*}
  We may write 
  $H^{-,\pi}_t = \sum_{k=1}^{N_t}\pi_{\tau_k}Z_k$, 
  where $\{N_t\}$ is a Poisson process with intensity $\eta(-1,-\vep)$
  and arrival sequence $\tau_{1}, \tau_2, \dots$, 
  independent of the sequence (of jump sizes)  
  $\{Z_k\}$ with probability distribution 
  $\eta(\cdot \cap (-1,-\vep))/\eta(-1,-\vep)$. 
  Then $\calE(H^{-,\pi})_t=\prod_{k=1}^{N_t}(1+\pi_{\tau_k}Z_k)$.
  Hence, 
  \begin{align*}
    \E\sup_{t\in [0,1]}\calE(H^\pi)_{t}^{-\alpha-\delta} &\leq
    \E\sup_{t\in [0,1]}\calE(H^{-,\pi})_{t}^{-\alpha-\delta}\\
    &=\E\left(\prod_{k=1}^{N_1}(1+\pi_{\tau_k}Z_k)\right)^{-\alpha-\delta}\\
    &\leq \E\left(\prod_{k=1}^{N_1}(1+Z_k)\right)^{-\alpha-\delta}\\
    &=\E e^{-(\alpha+\delta) \sum_{k=1}^{N_1}\log (1+Z_k)}\\
    &=\exp\{\eta(-1,-\vep)(\exp\{M(-\alpha-\delta)\}-1)\},
  \end{align*}
  where $M$ is the moment generating function of
  $\log(1+Z_1)$. Since
  \begin{align*}
    M(-\alpha-\delta) = E (1+Z_1)^{-\alpha-\delta} = 
    \eta(-1,-\vep)^{-1}\int_{-1}^{-\vep}
    (1+z)^{-\alpha-\delta} 
    \eta(dz) < \infty,
  \end{align*}
  the claim, for the case $n=1$, follows.

  For a general $n$ we may, with the similar notation as above, write
  \begin{align*}
    \calE(Z^\pi)_t=e^{\int_{0+}^t \pi_{s}^0 r_{s-}ds}
    \prod_{k=1}^n \calE(F^{\pi,k})_{t}\calE(G^{\pi,k})_{t}\calE(H^{\pi,k})_{t}.
  \end{align*}
  We know from the proof for the case $n=1$ that only the factors
  $\calE(H^{\pi,k})_{t}$ may cause problems with existence of moments.
  Using H\"older's inequality and following the arguments above we find 
  that 
  \begin{align*}
    \E \sup_{t\in [0,1]}\calE(H^{\pi,k})_{t}^{-n\alpha-\delta} < \infty
    \quad \text{for each } k,
  \end{align*}
  which follows from the assumptions on the L\'evy measures $\eta^k$,
  is sufficient to ensure that 
  $\E\sup_{t\in [0,1]}\calE(Z^\pi)_t^{-\alpha-\delta/n}<\infty$.
  This completes the proof.
\end{proof}

\begin{lem}\label{below}
  Let $M$ be a martingale and set $M^*_t=\sup_{s\in [0,t]}|M_s|$. 
  Then, for $\lambda > 0$,
  $P(M^*_t \geq x) \leq e^{-\lambda x} E e^{\lambda |M_t|}$.
  Moreover, if $E e^{\lambda M_t} < \infty$ for all $\lambda > 0$, 
  then $E e^{\lambda M^*_t} < \infty$ for all $\lambda > 0$.
\end{lem}
\begin{proof}
  Take $\lambda > 0$. Without loss of generality we may assume that 
  $Ee^{\lambda |M_t|} < \infty$. Note that since $x \mapsto e^{\lambda |x|}$
  is convex, $e^{\lambda |M_t|}$ is a submartingale. 
  Let $\tau = \min\{t, \inf\{s > 0: |M_s| > x\}\}$. Then, 
  \begin{align*}
    E e^{\lambda |M_t|} &\geq E e^{\lambda |M_\tau|}
    = E e^{\lambda |M_\tau|} I_{\{M^*_t \geq x\}} 
    + E e^{\lambda |M_\tau|} I_{\{M^*_t < x\}}\\
    &\geq e^{\lambda x} P(M^*_t \geq x) 
    + E e^{\lambda |M_t|} I_{\{M^*_t < x\}}. 
  \end{align*}
  Hence, 
  $P(M^*_t \geq x) \leq e^{-\lambda x} E e^{\lambda |M_t|} 
  I_{\{M^*_t \geq x\}} \leq e^{-\lambda x} E e^{\lambda |M_t|}$.
  For the last statement, take $\xi > \lambda > 0$. 
  Then 
  \begin{align*}
    E e^{\lambda M^*_t} &= \int_0^\infty P(e^{\lambda M^*_t}>x)dx\\ 
    &= 1 + \lambda\int_0^\infty e^{\lambda x} P(M^*_t > x)dx 
    \leq 1 + \lambda E e^{\xi |M_t|} \int_0^\infty
    e^{(\lambda-\xi)x}dx < \infty.
  \end{align*}
\end{proof}

\begin{proof}[Proof of Proposition \ref{thm:asoptimalconst}]
  Set $f_s=r_{s-}+\pi_{s}(\mu_{s-}-r_{s-})$ and $g_s=\pi_{s}\sigma_{s-}$
  and take $\beta \in (\alpha,\gamma)$.
  Then 
  \begin{align*}
    \calE(Z^\pi)_t^{-\beta} = \exp\left\{-\beta\left(
    \int_0^tf_sds+\int_0^tg_sdB_s-\frac{1}{2}\int_0^tg_s^2ds\right)\right\}.
  \end{align*}
  H\"older's inequality gives, with $1/p+1/q=1$ and $q$ small so that 
  $q\beta < \gamma$, 
  \begin{align*}
    \E\sup_{t\in [0,1]}\calE(Z^\pi)_t^{-\beta}\leq &
    \Big(\E\sup_{t\in [0,1]}\exp\Big\{
    -p\beta \int_0^tf_sds\Big\}\Big)^{1/p}\\
    &\Big(\E\sup_{t\in [0,1]}\underbrace{\exp\Big\{
    -q\beta\int_0^tg_sdB_s+\frac{q\beta}{2}\int_0^tg_s^2ds\Big\}}_{K_t}
    \Big)^{1/q}.
  \end{align*}
  Take $r\leq 2\gamma$ and note that $-r\int_0^tg_sdB_s$ is a continuous 
  local martingale and that
  \begin{align*}
    \E\exp\Big\{\frac{r^2}{2}\int_0^tg_s^2ds\Big\}<\infty.
  \end{align*}
  It follows from Theorem 41 on page 140 in \cite{P04} that $M = \{M_t\}$
  given by
  \begin{align*}
    M_t=\exp\Big\{-r\int_0^tg_sdB_s
    -\frac{r^2}{2}\int_0^tg_s^2ds\Big\}
  \end{align*}
  is a nonnegative martingale.
  Hence, $K=\{K_t\}$ is a submartingale so
  Theorem 20 on page 11 in \cite{P04} gives
  \begin{align}
    \E\sup_{t\in [0,1]}K_t 
    &\leq p^q \sup_{t\in [0,1]}
    \E\exp\Big\{
    -q\beta\int_0^tg_sdB_s+\frac{q\beta}{2}\int_0^tg_s^2ds\Big\}
    \nonumber \\
    &\leq p^q
    \E\exp\Big\{
      -q\beta\int_0^1g_sdB_s+\frac{q\beta}{2}\int_0^1g_s^2ds\Big\} 
      \label{longexpr}
  \end{align}
  To show that the expectation in \eqref{longexpr} is finite we 
  set $r:=q\beta < \gamma$ and note that
  \begin{align*}
    \exp\Big\{-r\int_0^1g_sdB_s+\frac{r}{2}\int_0^1g_s^2ds\Big\}
    =& \Big(\exp\Big\{
    -2r\int_0^1g_sdB_s-2r^2\int_0^1g_s^2ds\Big\}\Big)^{1/2}\\
    &\Big(\exp\Big\{
      (2r^2+r)\int_0^1g_s^2ds\Big\}\Big)^{1/2}.
  \end{align*}
  Hence, with $\gamma=q\beta$ for $q$ sufficiently small,
  the Cauchy-Schwarz inequality yields that the 
  expectation in \eqref{longexpr} is finite.
\end{proof}

\begin{proof}[Proof of Proposition \ref{thm:asoptimal}]
  The process $M$ given by $M_t = \alpha\int_{0+}^t \pi_{s}\sigma_{s-}dB_s$
  is a continuous local martingale if \eqref{eq:comcond} holds. 
  The Novikov condition \eqref{eq:comcond}
  and Theorem 41, p.~140, in \cite{P04} guarantee that $\calE(M)$ given by
  \begin{align*}
    \calE(M)_t = \exp\left\{\alpha \int_{0+}^t \pi_{s}\sigma_{s-}dB_s 
    - \frac{\alpha^2}{2} \int_{0+}^t \pi_{s}^2 \sigma_{s-}^2 ds\right\}
  \end{align*}
  is a uniformly integrable martingale. Hence, for every $\pi \in \Pi$,
  the measure $Q_\pi$ given by
  \begin{align*}
    E\left(\frac{dQ_\pi}{d\Prob} \,\Big|\, \mathcal{F}_t \right) = \calE(M)_t
  \end{align*}
  is a probability measure (equivalent to $\Prob$).
  Therefore we may write
  \begin{align*}
    \E \calE(Z^\pi)_{t}^{-\alpha} 
    = \E_{Q_\pi}\exp\left\{\alpha\int_{0}^t\left(
    -(1-\pi_{s})r_{s} -\pi_{s}\mu_{s}+
    \frac{1+\alpha}{2}\pi_{s}^2\sigma_{s}^2\right) ds\right\}.
  \end{align*}
  Hence, minimizing $\E \calE(Z^\pi)_{t}^{-\alpha}$ with respect to $\pi$
  is equivalent to minimizing the integrand on the right-hand side above.
  Since $\pi \mapsto -(1-\pi)r- \pi \mu + \frac{1+\alpha}{2}\pi^2\sigma^2$ 
  has a unique minimum at
  $\pi^* = \frac{\mu - r}{(1+\alpha)\sigma^2}$ the claim follows.
\end{proof}


\begin{thebibliography}{100}
\bibitem{AP07}
  L. Andersen and V. Piterbarg,
  \emph{Moment explosions in stochastic volatility models},
  Finance Stochast. {\bf 11} (2007) 29-50.
\bibitem{A04}
  D. Applebaum,
  \emph{L\'evy Processes and Stochastic Calculus},
  Cambridge Studies in Advanced Mathematics (No. 93),
  Cambridge University Press, 2004.
\bibitem{BGT87}
  N. Bingham, C. Goldie and J. Teugels,
  \emph{Regular Variation}, 
  Cambridge University Press, Cambridge, 1987.  
\bibitem{FS06}
W. Fleming and H. M. Soner,
\emph{Controlled Markov Processes and Viscosity Solutions},
2nd Ed, Springer, 2006.
\bibitem{GG02}
  J. Gaier and P. Grandits,
  \emph{Ruin probabilities in the presence of regularly varying tails
    and optimal investment},
  Insurance Math. Econom., \textbf{30} (2002) 211-217.
\bibitem{HP00}
  C. Hipp and M. Plum,
  \emph{Optimal investment for insurers}, 
  Insurance Math. Econom., \textbf{27} (2000) 215-228.
\bibitem{HL07}
  H. Hult and F. Lindskog,
  \emph{Extremal behavior of stochastic integrals driven by regularly
  varying L\'evy processes},
  Ann. Probab. \textbf{35} (2007) 309-339. 
\bibitem{KS02}
  J. Kallsen and A.~N. Shiryaev,
  \emph{The cumulant process and Esscher's change of measure},
  Finance Stochast. \textbf{6} (2002) 397-428.
\bibitem{KK08}
  C.~Kl\"uppelberg and R.~Kostadinova,
  \emph{Integrated risk models with exponential L\'evy investment}
  Insurance Math. Econom., \textbf{42} (2008) 560-577.
\bibitem{KM05}
  D.~G. Konstantinides and T. Mikosch,
  \emph{Large deviations and ruin probabilities for solutions to 
    stochastic recurrence equations with heavy-tailed innovations},
  Ann. Probab. \textbf{33} (2005) 1992-2035. 
\bibitem{P04}
  P. Protter,
  \emph{Stochastic Integration and Differential Equations},
  second edition.
  Springer, New York, 2004.
\bibitem{S99}
  Sato, K.-I., 
  \emph{L\'evy processes and infinitely divisible distributions},
  Cambridge University Press, 1999. 
\bibitem{S05}
  Schmidli, H.,
  \emph{On optimal investment and subexponential claims},
  Insurance Math. Econom., \textbf{36} (2005) 25-35.
\end{thebibliography}
\end{document}